\documentclass[conference]{IEEEtran}
\usepackage[T1]{fontenc}
\IEEEoverridecommandlockouts

\usepackage{graphicx}
\usepackage{physics}
\usepackage{amssymb}
\usepackage{amsthm}
\usepackage{hyperref} 
\usepackage{algorithm}
\usepackage{xcolor}
\usepackage{dsfont}
\usepackage{algpseudocode}
\usepackage[parfill]{parskip}
\usepackage[normalem]{ulem}
\usepackage{etoolbox}
\usepackage{subcaption}
\usepackage{mathtools}
\usepackage{braket}
\usepackage{tikz}
\usepackage{enumitem}
\usetikzlibrary{quantikz2, positioning,fit,backgrounds}

\apptocmd{\thebibliography}{\small}{}{}

\definecolor{GRAblue}{HTML}{045275} 
\definecolor{GRgreen}{HTML}{089099}
\definecolor{GRlightgreen}{HTML}{7CCBA2}
\definecolor{GRorange}{HTML}{F0746E}
\definecolor{GRpink}{HTML}{DC3977} 

\MakeRobust{\Call}

\newtheorem{theorem}{Theorem}
\newtheorem{definition}[theorem]{Definition}
\newtheorem{proposition}[theorem]{Proposition}

\theoremstyle{definition}

\newtheorem{example}[theorem]{Example}

\begin{document}

\title{Approximate Sparse State Preparation with the Grover--Rudolph Algorithm}

\author{
\IEEEauthorblockN{
Debora Ramacciotti\IEEEauthorrefmark{1},
Martin Steinbach\IEEEauthorrefmark{1},
Bence Temesi\IEEEauthorrefmark{1},
Andreea-Iulia Lefterovici\IEEEauthorrefmark{1},
Antonio F.\ Rotundo\IEEEauthorrefmark{2}
}

\IEEEauthorblockA{
\IEEEauthorrefmark{1}Institut f\"{u}r Theoretische Physik\\
Leibniz Universit\"{a}t Hannover, Hannover, Germany\\
Email: $\{$debora.ramacciotti, martin.steinbach, bence.temesi, andreea.lefterovici$\}$@itp.uni-hannover.de
}

\IEEEauthorblockA{
\IEEEauthorrefmark{2}Fermioniq B.V.\\
Amsterdam, The Netherlands\\
Email: af.rotundo@gmail.com
}
}

\maketitle

%\begin{abstract}

% Efficient quantum state preparation is a requirement for implementing quantum algorithms on early fault-tolerant quantum hardware. %In this context, s
% Sparse states are particularly of interest as they appear in applications ranging from simulation to optimization.
% The goal of this paper is to study approximate quantum state preparation for sparse states. 
% First, we introduce a hybrid strategy that combines two different approaches:
% (i) we decompose the rotation gates one by one or (ii) we group the rotation gates into blocks of controlled uniform rotations. We select either (i) or (ii) depending on which one is more cost efficient.
% This small optimization strategy already brings a $90\%$ reduction in the number of gates for states with $10^{-4}$ sparsity percentage.
% Second, we allow merging of gates that would add a controlled error in the prepared state.
% We derive %explicit 
% lower bounds on the resulting overlap with the target state. These bounds can be efficiently evaluated classically and used to determine the optimal gate merging.
% %Our results show $\dots$ \textcolor{red}{waiting for the new plots}
% \end{abstract}

\begin{abstract}
Sparse quantum state preparation is a common subroutine in quantum algorithms, where classical data with few nonzero entries must be loaded into a quantum state. 
In this work, we consider the Grover--Rudolph algorithm, which has recently been shown to efficiently prepare sparse states, and we propose two improvements.
First, we extend an existing gate-merging procedure by allowing rotations to merge with virtual zero-angle gates on unreachable branches of the preparation tree, reducing the number of CNOTs and control qubits. 
Second, we introduce an approximate variant in which rotations with similar but not identical angles are merged at the cost of a small, controllable error in the prepared state. 
We derive a classically computable estimate of the resulting overlap with the target state, which is used to guide the merging decisions.
    %The goal of this paper is to optimize the Grover-Rudolph state preparation algorithm for sparse states. First, we introduce an exact optimization that, by merging gates, achieves a $\mathbf{99.9\%}$ reduction in the CNOT count for states with a sparsity percentage of $\mathbf{D\approx 10^{-5}}$. We then introduce a small error in the preparation to further reduce the number of CNOTs, while maintaining an overlap above a certain threshold with the desired state. Thanks to this approximation, we achieve a further $\mathbf{20-30\%}$ reduction in the CNOT count compared to the previous approach.
\end{abstract}

\begin{IEEEkeywords}
sparse state preparation, approximate sparse state preparation, grover-rudolph
\end{IEEEkeywords}

%\tableofcontents

\section{Introduction}
Given a vector $\psi$, the goal of quantum state preparation (QSP) is to build a unitary $U$ such that  $U\ket{0}=\ket{\psi}$, where $\ket{\psi}$ is a quantum state whose amplitudes are equal to $\psi$ up to normalization. 
Many quantum algorithms begin by loading classical data into the amplitudes of a quantum state, making QSP an essential subroutine in quantum computing.

The complexity of preparing an unstructured quantum state scales exponentially in the number of qubits \cite{knill1995approximation, zalka1998simulating, kaye2001quantum,grover2002creating, sun2023asymptotically}, making preparation infeasible for all but small systems.
Therefore, many QSP algorithms exploit structure in the target state to achieve lower complexity.
In this paper, we focus on sparse states, i.e.,\ states where most computational basis amplitudes are zero.
More precisely, let $n$ be the number of qubits and $d$ the number of nonzero amplitudes; we focus on the regime $d\ll 2^n$. 

Many QSP algorithms tailored to sparse states have been proposed in recent years \cite{gleinig2021efficient, malvetti2021quantum, de2022double, mozafari2022efficient, zhang2022quantum, ramacciotti2024simple, mao2024optimalcircuitsizesparse, rupprecht2026sparse}.
It has been shown that sparse QSP admits a trade-off between gate count (time complexity) and ancilla count (space complexity) \cite{zhang2022quantum, li2024nearly}.
For example, the optimal circuit depth, $\Theta(\log(nd))$, can be achieved using $\mathcal{O}(nd\log d)$ ancillas \cite{zhang2022quantum}.
Here, we focus on the opposite regime, with only $\Theta(n)$ ancillas, which is the most relevant for early fault-tolerant quantum computers \cite{preskill2025beyond, eisert2025mind}.

Among $\Theta(n)$-ancilla algorithms, the Grover--Rudolph \cite{grover2002creating} algorithm offers a particularly simple baseline.
Although originally designed to load unstructured states or discretized efficiently-integrable distributions, it has recently been shown to be efficient for preparing sparse states \cite{ramacciotti2024simple}.
The algorithm prepares a state by applying layers of controlled rotations; for sparse states, only $\mathcal{O}(nd)$ such rotations are required.
Despite this favourable scaling, the resulting circuits remain expensive in absolute terms.
We improve this algorithm along two directions.

First, we develop an optimization routine that reduces the number of gates required to implement the Grover--Rudolph algorithm.
The core idea, already explored in \cite{ramacciotti2024simple}, is to merge rotations whose control patterns are neighboring in Hamming distance, provided that doing so leaves the prepared state unchanged.
This reduces both the number of controlled rotations and the number of control qubits. 
The main new contribution of this paper is that, compared to \cite{ramacciotti2024simple}, we %additionally 
allow merges between an isolated rotation and a neighboring virtual zero-angle rotation on an unreachable branch.
This does not reduce the number of rotations, but it does reduce the number of control qubits; effectively this is a control stripping move.
Empirically, we find that this type of merge is very common for sparse states and significantly reduces the number of CNOT gates.
For example, for random states with $d/2^n \approx 10^{-5}$, the full procedure reduces the CNOT count by $90\%$ compared to the unoptimized Grover--Rudolph circuit.

Second, we extend the approximate QSP algorithm of \cite{marin2023quantum} to sparse states. 
In approximate QSP, the gate count is reduced by allowing for a small error in the prepared state \cite{zulehner2020approximation, Rofougaran_2025, RINDELL2023128860, zylberman2025efficient, 10190145}. 
In particular, \cite{marin2023quantum} shows that for amplitudes sampled from a smooth function on a uniform grid, the CNOT count can be reduced to a value asymptotically independent of $n$ at the cost of a small square-root infidelity.
Sparse states do not satisfy this smoothness condition, as their amplitudes correspond to highly discontinuous functions on the uniform grid.
Therefore, they require a different approximation mechanism based on the support structure. 
We fill in this gap by allowing bounded approximation errors in the merge optimisation introduced above. 
To do so, we derive an estimate of the overlap between the exact target state and the running approximate state at each step of the merging procedure, which can be efficiently evaluated classically and used to decide which merges to perform. 
Empirically, we find that this leads to an additional $20$--$30\%$ reduction in the CNOT count. 

\emph{Related work.} The approach closest to ours is that of Zulehner et al.~\cite{zulehner2020approximation}: a quantum state is approximated by reducing its decision-diagram representation before synthesis, pruning low-contribution nodes while controlling the fidelity loss. Our construction differs in that it acts directly on the rotation angles of the Grover--Rudolph circuit rather than on a separate state representation.

The paper is organized as follows. In Sec.~\ref{sec:preliminaries} we review the Grover--Rudolph algorithm and the merging method from \cite{ramacciotti2024simple}. 
In Sec.~\ref{sec:exactoptimization} we introduce an exact optimisation method and compare its CNOT count to the standard Grover--Rudolph implementation. 
In Sec.~\ref{sec:approx} we present the approximate algorithm, derive the overlap estimate, and compare its CNOT count to that of the exact algorithm of Sec.~\ref{sec:exactoptimization}.
\section{Preliminaries}
\label{sec:preliminaries}
\subsection{Grover--Rudolph Algorithm}\label{subsection:GroverRudolph}

\begin{figure}
    \centering
    \includegraphics[width=0.7\linewidth]{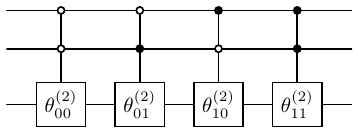}
    \caption{Circuit diagram of a uniformly controlled y-rotation on three qubits. We indicate the $R_y$ rotation just by the applied angle to simplify the notation.}
    \label{circuit:uniformly_controlled_rotation}
\end{figure}
%Let $\psi \in \mathbb{C}^{N}$ be a classical vector.
%The goal of a state preparation algorithm is to construct a unitary $U_{\psi}$, such that $U_{\psi} \ket{0} = \ket{\psi}$, where $\ket{\psi}$ is a quantum state with amplitudes equal to $\psi$, up to the normalisation constant.
%One method to prepare such a state is the Grover-Rudolph algorithm, which constructs a series of coarse-grained versions of $\psi$, recursively preparing them using controlled rotations.
The Grover--Rudolph algorithm \cite{grover2002creating} prepares the state $\ket{\psi}$ by recursively constructing a series of coarse-grained versions of $\ket{\psi}$ using controlled rotations.
In this paper, we only consider the real positive case of $\psi \in \mathbb{R}_+^{N}$.
More precisely, we want to prepare
\begin{equation*}\label{eq:state_to_prepare}
    U_{\psi}\ket{0}= \frac{1}{\norm{\psi}}\sum_{i_1\dots i_n}\psi_{i_1\dots i_n}\ket{i_1 \dots i_n},
\end{equation*}
where the indices $i_k$ take values in $\{0,1\}$, $\norm{\psi}$ is the 2-norm of the vector, and $N=2^n$ is the total number of summation terms. We assume \(\norm{\psi}=1\) in the rest of the paper.

The coarse-grained states, denoted $\psi^{(k)}$ for $k=1,\dots, n$, have components
\begin{equation}\label{eq:coarse_grained_state}
\psi_{i_1 \dots i_{k}}^{(k)} = 
    \sqrt{(\psi^{(k+1)}_{i_1 \dots i_{k} 0})^2 + (\psi^{(k+1)}_{i_1 \dots i_{k} 1})^2}\,,
\end{equation}
for $k<n$, while $\psi^{(n)}\equiv \psi$.
The superscript $(k)$ keeps track of the number of qubits required to encode $\ket{\psi^{(k)}}$, the quantum state having amplitudes $\psi_{i_1 \dots i_{k}}^{(k)}$.
Notice that $\psi^{(0)}=1$ by normalization.
The state $\ket{\psi^{(k)}}$ is prepared starting from $\ket{\psi^{(k-1)}}$ by appending an ancilla qubit $\ket{0}$ and performing a rotation of the form
\begin{equation}
  \begin{aligned}\label{eq:recursion}
    &\ket{{i_1 \dots i_k}}\ket{0} \rightarrow \\
    &\ket{{i_1 \dots i_k}}  
    \Bigl(\cos{\frac{\theta_{i_1 \dots i_k}^{(k)}}{2}}\ket{0} +  
    \sin{\frac{\theta_{i_1 \dots i_k}^{(k)}}{2}} \ket{1} 
    \Bigr).
\end{aligned}  
\end{equation}
\begin{figure}
    \centering
    \includegraphics[width=\linewidth]{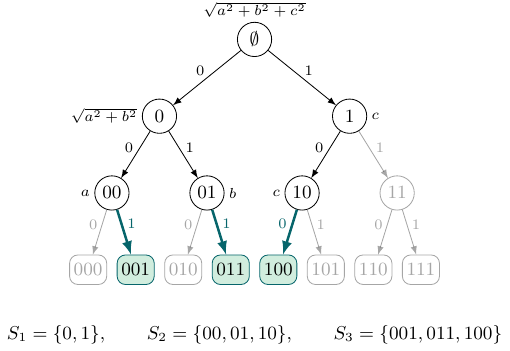}
\caption{
Preparation tree for the state \(\ket{\psi}=a\ket{001}+b\ket{011}+c\ket{100}\).
The highlighted leaves are the nonzero basis states of the target state, while faded nodes indicate zero-amplitude branches. The non-faded nodes correspond to the nonzero prefix sets \(S_k\). The coarse-grained state $\ket{\psi^{(2)}}$, for instance, is given by $a\ket{00} + b\ket{01} + c \ket{10}$.
}
\label{fig:prep_tree_sparse_example}
\end{figure}
The angles should be chosen such that
\begin{equation}\label{eq:angles_phases}
    \theta_{i_1 \dots i_k}^{(k)} = 2 \arccos{\frac{\psi^{(k+1)}_{i_1 \dots i_{k} 0}}{\psi^{(k)}_{i_1 \dots i_k}}}\,,
\end{equation}
where if \(\psi^{(k)}_{i_1 \dots i_k}=0\) the angle may be set to \(0\) by convention.
The superscript of the angles $(k)$ corresponds to the number of control qubits in the rotation, and the subscript $i_1 \dots i_k$ to the control pattern.

In the case of sparse vectors, we assume that we know the number of nonzero entries in $\psi$ and their locations. 
The rotation angles can be found by exploiting the sparsity of the vector, as shown in Algorithm 3 in \cite{ramacciotti2024simple}.
Since each $\psi^{(k)}$ has at most $d$ nonzero elements, we find that each unitary implementing $U$ has at most $d$ nonzero entries.

For dense states, instead, all angles in \eqref{eq:angles_phases} are generally nonzero.
In this case, each layer of gates preparing the coarse grained states consists of a so-called uniformly controlled rotation ($\mathrm{UCR}^{(k)}$) over the $y$-axis \cite{mottonen2004transformation}. 
A $\mathrm{UCR}^{(k)}$ is a series of multiple rotations that act on the $k$-th qubit, controlled by all possible combinations of controls on the $k$ qubits, as depicted in Fig.~\ref{circuit:uniformly_controlled_rotation}.
The structure of a $\mathrm{UCR}^{(k)}$ allows it to be decomposed very efficiently, using only $2^{k}$ CNOTs. Furthermore, a state can be prepared with only $2^n$ CNOTs by employing uniformly controlled rotations \cite{mottonen2004transformation, Shende_2006}, using a different implementation from the Grover-Rudolph one.

Eq. \eqref{eq:coarse_grained_state} defines a preparation tree. 
As an example, consider the state
\begin{equation}\label{eq:ex_state}
\ket{\psi}=a\ket{001}+b\ket{011}+c\ket{100},
\end{equation}
whose preparation tree is shown in Fig.~\ref{fig:prep_tree_sparse_example}.

The probabilities related to each coarse-grained state can be written both in terms of the amplitudes and in terms of the angles.
For a bit string $i_1\dots i_k\in\{0,1\}^k$, let
$P_{i_1\dots i_k}^{(k)} = (\psi_{i_1\dots i_k}^{(k)})^2$
denote the probability of reaching the node labelled by $i_1\dots i_k$ at depth $k$ in the preparation tree. 
Equivalently, using \eqref{eq:recursion}, the same probability is obtained from the rotation angles along the path:
\begin{equation*}\label{eq:prefix_probability_angles}
    P_{i_1\dots i_k}^{(k)}
    =
    \prod_{l=1}^{k}
    f\!\left(\theta_{i_1\dots i_{l-1}}^{(l-1)},\, i_l\right),
\end{equation*}
where
\begin{equation*}
    f(\theta,0)=\cos^2\!\frac{\theta}{2},
    \qquad
    f(\theta,1)=\sin^2\!\frac{\theta}{2}.
\end{equation*}

In the following, we use the notation: for the $k$-th layer, $S_k \subseteq \{0,1\}^k$ denotes the set of computational basis prefixes that occur with nonzero amplitude before the rotations of that layer are applied. Equivalently, $s \in S_k$ if and only if the coarse-grained amplitude $\psi_s^{(k)}$ is nonzero.

\subsection{Merging Optimization}
\label{subsec:merging}
One can use a simple optimization strategy to reduce the number of controlled rotations used in the Grover--Rudolph algorithm.
Consecutive gates that have the same rotation angles and whose controls differ only by one bit-flip can be merged without changing the prepared state.

We show a simple example in Fig.~\ref{fig:merging_example}.
The first gate is conditioned on `11' and performs a rotation with an angle $\theta$, while the second gate is conditioned on `10' and executes a rotation with the same angle $\theta$. Given that the gates differ by only one control and have the same rotation angles, they can be combined into a single gate, controlled on the first qubit being in `1'. Let us introduce the notation of the `e' control, which simply means that, in that position, no control is applied. Thus, the merging of `11' and `10' would be indicated by `1e'.

This leads not only to a reduction in the number of gates, but also to gates with fewer control qubits. 
The implementation is explained in detail in \cite{ramacciotti2024simple}, and has a classical cost of $O(dn^2\log d)$.

In Sec.~\ref {sec:exactoptimization}, we strengthen this optimization procedure by exploiting the support structure of the preparation tree, which allows us to remove some controls.

\begin{figure}
    \centering
\includegraphics[scale=0.85]{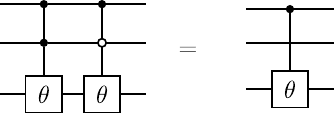}
    \caption{Example of optimization procedure. Here, $\theta^{(2)}_{11} = \theta^{(2)}_{10} = \theta^{(2)}_{1e} = \theta$.}
    \label{fig:merging_example}
\end{figure}

\section{Exact Optimization}
\label{sec:exactoptimization}

In this section, we introduce an exact optimization strategy that reduces the number of required CNOTs by merging some gates and combining this %decomposition 
approach with the one of uniformly controlled rotations (UCR). 

\subsection{Exact merges}
\label{subsec:merge&strip}
We will allow two kinds of merges:
\begin{enumerate}
    \item[(i)] \textit{neighboring merges}: a merging of two gates of equal angle on the same layer, whose control patterns differ only in one bit (see Sec. \ref{subsec:merging}) %and with the same angles, as explained in Sec. \ref{subsec:merging};
    \item[(ii)] \textit{Control stripping}: a stripping of one control, that is, a merge with a \textit{virtual zero gate}\footnote{A virtual zero gate is a gate with a zero rotation angle that will not be applied in the circuit in practice.}, %i.e., 
    on an unreachable branch of the preparation tree.
\end{enumerate}

\begin{example}[Stripping]
    %Let us take, for instance, the preparation of the state
    Consider again the state \eqref{eq:ex_state}
    whose preparation tree is shown in Fig. \ref{fig:prep_tree_sparse_example}. A possible stripping of a control can be performed on the node $x=01$, which %would 
    becomes $z = e1$, and keeps its old angle $\theta_{01}^{(2)}$. %Note how 
    This is the same as merging $x$ with the node $x'=11$, %which is 
    located on an unreachable branch on the tree, i.e., on a branch with zero amplitude.  
\end{example}
This kind of merging, with a gate on an unreachable branch, is always exact.
To formalize this, we introduce the following definition.
\begin{definition}
    For a control pattern $x \in \{0, 1, e\}^k$, the set of bitstrings consistent with $x$ is $$B(x) := \{b \in \{0, 1\}^k : b_t = x_t,\;\forall\, t \;\textup{with}\; x_t \neq e \},$$ %Its cardinality is 
    where $|B(x)| = 2^m$, %with 
    and $m = |\{t : x_t = e\}|$.
\end{definition}
In simple terms, $B(x)$ is the set given by replacing the empty control $e$ with either $0$ or $1$. %, taking into account all the possibilities.
For instance, $B(0e) = \{ 00, 01 \}$.

 Let $x\in\{0,1,e\}^k$ be the control pattern of a gate with angle $\theta^{(k)}_x$, and let $z$ be obtained from $x$ by replacing one non-empty control with $e$; that is, $z$ represents a possible merging candidate. 
 We denote by $x'$ the complementary pattern obtained by flipping that control and obtain %Then
\begin{equation*}
    B(z)=B(x)\cup B(x').
\end{equation*}
If the branch of $x'$ is unreachable, that is,
\begin{equation} \label{eq:support}
    S_k\cap B(x')=\varnothing,
\end{equation}
then performing 
the control stripping represented by $z$, without changing the angle $\theta^{(k)}_x$, leaves the prepared state unchanged.

In fact, the coarse-grained state at layer $k$ is
\begin{equation*}
    \ket{\psi^{(k)}} = \sum_{b \in \{0,1\}^k} \alpha_b \ket{b},
\end{equation*}
where $\alpha_b := \langle b | \psi^{(k)} \rangle$.
Thus, when $b \in B(x)$, nothing changes since the angle stays the same, and when $b \in B(x')$, the value of the angle is not relevant, since $\alpha_b = 0$.

\subsection{The Algorithm}
The exact optimization at each layer is divided into three steps:
\begin{enumerate}
    \item[(1)] For each gate, strip controls \emph{sequentially} (as explained below) as long as \eqref{eq:support} holds for the current key;
    \item[(2)] Merge neighboring gates with the same angle, until no other neighbor is found;
    \item[(3)] Count the %number of 
    CNOTs of each controlled rotation singularly. Apply the merging optimization only if the total number of CNOTs for the single rotations is smaller than $2^{k}$ (the number of CNOTs %gates 
    for a $\mathrm{UCR}^{(k)}$).
\end{enumerate}

We computed the number of CNOTs for a single rotation following \cite{vale2023decomposition}:
\begin{equation*}
N_{\text{CNOT}}^{\text{single}} =
\begin{cases}
0, & \text{if } N_{\text{ctrl}} = 0, \\[6pt]
2, & \text{if } N_{\text{ctrl}} = 1, \\[6pt]
16\,N_{\text{ctrl}} - 24, & \text{if } N_{\text{ctrl}} > 1~,
\end{cases}
\end{equation*}
where $N_{\text{ctrl}}\leq k$ is the number of controls. %and it is smaller than or equal to $k$.

Note that it is not necessary to reiterate this procedure, that is, performing again step (1) after step (2), since a pair merge cannot create a new control stripping opportunity: to see this, let $l$ and $r$ be arbitrary bit strings, and suppose that the rotation gate corresponding to the control patterns
    $x = l0r$, $y = l1r$,
have the same angle. Then, $x$ and $y$ can be merged into $z=ler$. Then,
    $B(z) = B(x) \cup B(y)$. \newline
% If $z$ is strippable from a control again, then there is a newly added branch from $z$ that is unreachable, but it
% %But then That branch 
% also must be unreachable from $x$ and $y$, i.e.\ the stripping could have already happened. %before.
If $z$ is strippable from a control again, then the complementary region added by this stripping has zero support. This
region is the union of the corresponding complementary regions for \(x\) and
\(y\). Hence, each of these complementary regions already had zero support
before the neighboring merge, and the corresponding stripping could already
have been performed in step (1).

However, it is not possible to compute \textit{all} the possible controls that can be stripped, and strip them all at once in step (1). We need to do this \textit{sequentially}. In fact, let us suppose that two zero support strippings are possible from $x=001$:
\begin{align*}
    &z_1 = e01, \qquad B(z_1) = \{001, 101\}, \\
    &z_2 = 0e1, \qquad B(z_2) = \{001, 011\},
\end{align*}
that is, both $101$ and $011$ are unreachable. If we want to strip them at the same time, then we %would 
get 
\begin{equation*}
    z = ee1, \qquad B(z) = \{001,011,101,111\},
\end{equation*}
which would imply that $111$ is also unreachable, which is generally not true.

\subsection{Classical Complexity}
As mentioned in Sec.~\ref{subsec:merging}, step (2) has a classical complexity $\mathcal{O}(d n^2 \log d)$, while step (3) only requires a constant number $\mathcal{O}(1)$ of comparisons between CNOT counts and the UCR count at each layer. Let's %now 
estimate the cost of step (1).

At layer $k$, there are at most $d$ nonzero rotation gates, and, for each gate, we check at most $k$ non-empty controls. Moreover, the stripping is performed sequentially: after one accepted removal, the current key changes, and the next support condition must be recomputed. In the worst case, this leads to $\mathcal{O}(k^2)$.
checks for each gate.

In the sparse regime, the support set $S_k$ contains at most $d$ elements. 
Checking whether
\(S_k\cap B(x')=\varnothing\) is done by scanning the sparse support and
testing control pattern compatibility. Since the elements of \(S_k\) are
\(k\)-bit strings, the cost of testing \eqref{eq:support} amounts to \(\mathcal{O}(dk)\). Taking into account the sequential procedure, the stripping cost for one gate is
\(\mathcal{O}(dk^3)\), and, since there are at most \(d\) gates at layer \(k\), the
total stripping cost at layer \(k\) is $\mathcal{O}(d^2 k^3)$.

Summing over all layers gives $\mathcal{O}(d^2 n^4$).
%\begin{equation*}
%    \sum_{k=1}^{n}\mathcal{O}(d^2 k^2)
%    =
%    \mathcal{O}(d^2 n^3).
%\end{equation*}
Therefore, the overall classical complexity of the exact optimization is $\mathcal{O}(d^2 n^4 + dn^2\log d).$
%\begin{equation*}
%    \mathcal{O}(d^2 n^3 + dn^2\log d).
%\end{equation*}

In practice, this estimate is often pessimistic. Many support checks terminate immediately, and repeated support queries can be stored. Nevertheless, the routine remains polynomial in both $d$ and $n$.

\subsection{Numerical Results}
First, we compared the CNOT reduction related to the merges. Fig.~\ref{fig:before_after_merge} shows the ratio between the CNOT count after the merging
procedure and the CNOT count without merges, both computed using the single-rotation decomposition.
When the sparsity percentage is about $10^{-3}$, we already achieve a CNOT reduction of about $50\%$, while for sparser states, e.g., $D = 10^{-5}$, the CNOT count is reduced by $90\%$. The increasing behavior of the number of merges with the sparsity is given by the fact that the sparser the state, the more branches of the preparation tree are unreachable, and then the more control strippings are possible. 

\begin{figure}
    \centering
    \includegraphics[width=\linewidth]{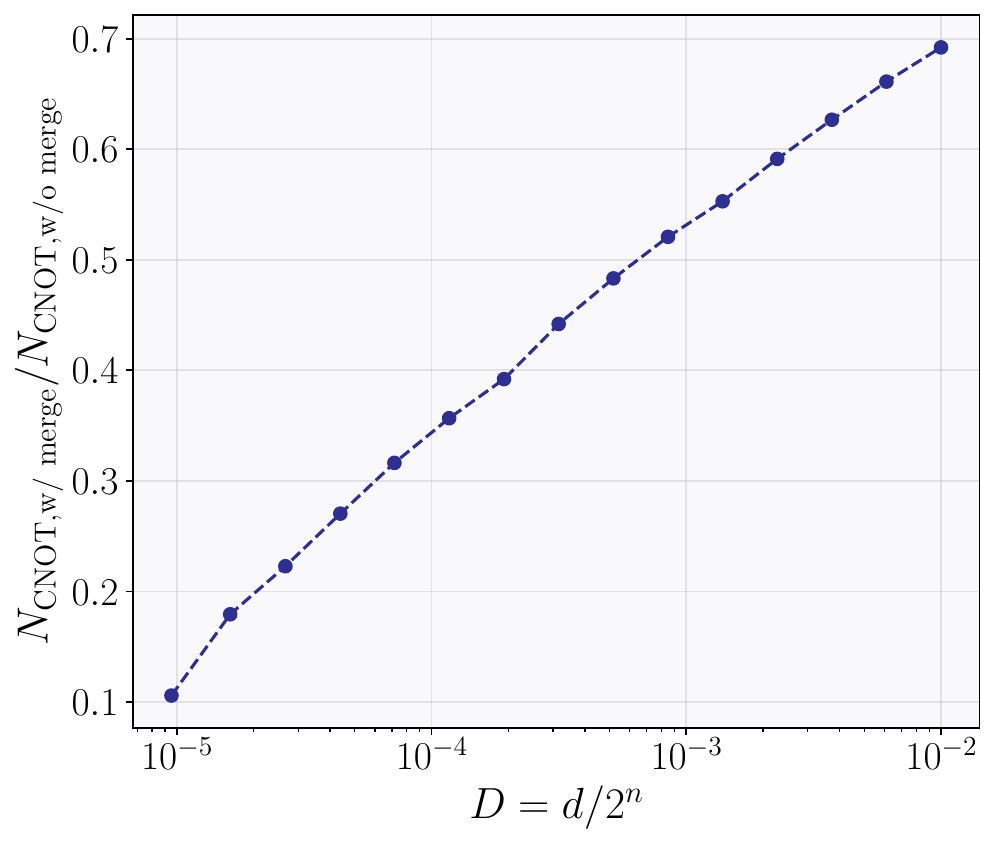}
    \caption{CNOT reduction due to the exact merges. We plot the ratio of the CNOT count after the merging procedure to that before the merging procedure, using the single-rotation decomposition in both cases, against the sparsity percentage $D =d / 2^n$. The number of qubits is fixed at $n=20$ qubits, and each point is averaged across 20 repetitions.}
    \label{fig:before_after_merge}
\end{figure}

Then, in Fig.~\ref{fig:cnot_comparison}, we compare the number of CNOTs required in the case of (i) uniform rotations (in blue), (ii) counting single rotations after the exact merges (in pink), and (iii) deciding layer by layer, either to apply the optimization or not (in green), as in the exact algorithm. 
For denser states ($D \approx 10^{-2}$), the exact algorithm approaches the UCR decomposition, while for sparser states, it approaches the one that decomposes rotations singularly. For a very sparse state, e.g., for $D \approx 10^{-5}$, we achieve a $99.9\%$ reduction compared to the UCR approach.

Note that our result could be further improved by using other decomposition algorithms that use ancillas, such as \cite{zindorf2025multi}, in which the authors claim that if we have all-to-all connectivity and $k \approx n$, there is an upper bound of $\approx 12 n$, where $n$ is the number of qubits.

\begin{figure}
    % D = d / 2^n explain sparsity percentage
    % change to #CNOT gates
    % use computer modern font and increase label size
    \centering
    \includegraphics[width=\linewidth]{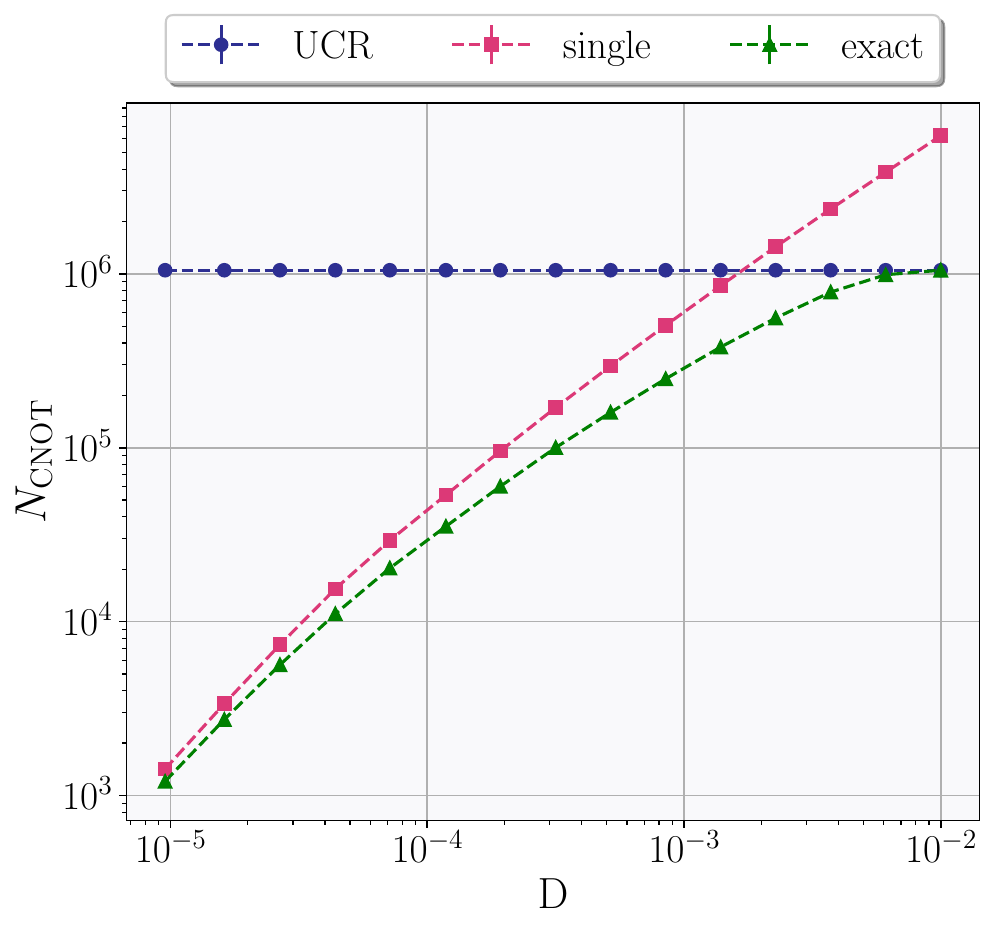}
    \caption{%Uniform rotation algorithm vs single merges depending on sparsity for $n=20$ number of qubits.
    Comparison between the uniform rotation algorithm (blue), single rotations after merges (pink), and the exact algorithm (green). We plot the CNOT gate count against the sparsity percentage $D =d / 2^n$, for simulation of $n=20$ qubits and average across 20 repetitions. %The number of CNOT gates is constant for the uniform rotation algorithm across all sparsity percentages, for the single rotations it increases linearly, and for the hybrid method it saturates to the same value as for the uniform method as D increases.
    }
    \label{fig:cnot_comparison}
\end{figure}
\section{The approximate Optimization}
\label{sec:approx}
%The purpose of the optimization is to further 
Our next goal is to reduce the number of gates by allowing a small error in the preparation process while maintaining an overlap with the target state above a prescribed threshold, which is specified as input to the algorithm. 

\subsection{Approximate merges}
In this section, we generalize the exact merges of Sec.~\ref{subsec:merge&strip} in order to allow for approximate merges. As before, we allow two types of moves:
\begin{enumerate}
    \item \textit{neighbor merges}: two rotations acting on the same layer whose control patterns differ in exactly one bit;
    \item \textit{Control stripping}: the removal of a non-empty control from an active gate, allowed only if the enlarged region contains only rotations of zero angles.
\end{enumerate}
In both cases, the new rotation angle of the resulting gate is, in general, different from the old angles and, therefore, introduces an approximation in the prepared state. The optimal value of the merged angle is given by \eqref{eq:theta_C}, and will be derived in the following section.

For control stripping, the flipped sibling pattern $x'$ is treated as a zero-angle gate. This remains true even if that sibling branch is reachable in the preparation tree and has a nonzero probability weight (the stripping would be exact otherwise). However, the move is allowed only if the enlarged region $B(x')$ contains only rotations of zero angles. 

\begin{example}[Not allowed Stripping]
Assume that at some layer the current active circuit contains the two gates
\[
0e0 \mapsto \theta_A,
\qquad
110 \mapsto \theta_B,
\qquad \theta_B\neq 0.
\]
Suppose we want to strip the first control from the gate \(0e0\). The resulting key would be $x = ee0$,
and the flipped sibling pattern is $x'=1e0$.
Then, the enlarged region corresponding to the new key \(ee0\) is
\[
B(ee0)=\{000,010,100,110\}.
\]
Inside this region, there is also another active gate, namely \(110\mapsto\theta_B\), thus the stripping is not allowed.
\end{example}
\subsection{The Algorithm}
Let us indicate with $\mathcal{F}$ the square-root fidelity, which, in the case of positive real states, equals the overlap
\begin{equation*}
    \mathcal{F}(\psi,\psi') = \langle \psi | \psi' \rangle,
\end{equation*}
and with $\mathcal{F}_{min}$ the minimum allowed overlap between the target and the prepared state, given as another input to the algorithm.\newline
The algorithm proceeds as follows:
\begin{enumerate}[label=(\arabic*)]
    \item Construct the exact Grover--Rudolph circuit and keep it as baseline data for evaluating the overlap estimator.
    \item Apply the improved exact merging procedure of Sec.~\ref{sec:exactoptimization} to obtain the active circuit.
    \item Divide the interval $[\,\mathcal{F}_{\min},1\,]$ into $M$ equally spaced thresholds
    \begin{equation*}
        \mathcal{F}^{(s)} = 1 - s\frac{1-\mathcal{F}_{\min}}{M},
        \qquad s=1,\dots,M.
    \end{equation*}
    \item For each threshold $\mathcal{F}^{(s)}$, generate all admissible candidate merges in the current active circuit.
    \item For every candidate, determine the set of baseline patterns absorbed by the candidate merge, compute the corresponding merged angle, and evaluate the post-merge overlap estimate.
    \item Order the admissible candidates by decreasing estimated post-merge overlap and process them greedily. Whenever a candidate is reached, recompute its overlap estimate using the current cluster state and accept it only if the updated estimate still satisfies $\mathcal{F}_{\mathrm{est}} \ge \mathcal{F}^{(s)}$.
    \item After the last threshold $\mathcal{F}_{\min}$ is reached, repeat the same greedy pass at fixed threshold $\mathcal{F}_{\min}$ until no further admissible merge exists.
\end{enumerate}

%Let us now clarify some of the steps.

In step (1) we build the \emph{baseline circuit}, which is the original exact Grover--Rudolph circuit, kept fixed throughout the algorithm. Its only purpose is to provide the probability weights used in the estimator. The \textit{active circuit} is the circuit currently being optimized by successive merges.

In step (3) we partition $[\,\mathcal{F}_{\min},1\,]$ into $M$ small intervals and optimize locally within each interval, because if we were to order them once and merge as long as our overlap estimation was less than (or equal to) than $\mathcal{F}_{\min}$, we would lose the opportunity to merge an already merged gate.\newline
\begin{example}
For example, suppose that a rotation has control pattern $11$ and that the neighboring branch $10$ is unreachable.
Then $11$ can be merged with a virtual zero-angle rotation on $10$, giving
$
    11 \longrightarrow 1e .
$
If the branch $0e$ is also unreachable, the newly created gate can then be stripped again,
$
    1e \longrightarrow ee .
$
However, this second move is not present in the original list of candidates, since the pattern $1e$ did not exist before the first merge.
Therefore, ordering the candidates only once may miss valid merges involving gates created during the optimization.
The local optimization over small overlap intervals avoids this problem by updating the dictionary and recomputing the available moves after each interval.
\end{example}

Step (5) requires a bookkeeping structure for repeated merges. To this end, each active gate carries the set of baseline control patterns that have been absorbed into it. We call this set its \emph{cluster}. Initially, every active gate forms its own cluster. Whenever two active gates are merged, their clusters are joined. In the case of a control stripping, the newly added sibling pattern is included in the cluster with baseline angle $0$. The role of these clusters in the overlap estimate will be made explicit in the following section.

\begin{example}[Clusters]
As an example, consider the second layer of the preparation tree in Fig. \ref{fig:prep_tree_sparse_example}, determined by the controls $00$, $01$, $10$, $11$. 
Initially, each active gate forms its own cluster, e.g., $C_{00} = \{00\}$.
If $10$ and $11$ are merged into the active gate $1e$, the new active cluster is
\[
    C_{1e}=\{10,11\}.
\]
If later the gate $1e$ is merged with $0e$ into $ee$, then the new cluster is obtained by joining the previous ones:
\[
    C_{ee}=C_{1e}\cup C_{0e}=\{00,01,10,11\}.
\]
\end{example}

In step (6), the admissible candidates are ordered by their estimated post-merge overlap. However, this ordering is only provisional: once a merge is accepted, the active circuit changes, and the overlap estimate of the remaining candidates may also change. For this reason, each candidate is re-evaluated when it is encountered during the greedy sweep and is accepted only if the updated estimate still satisfies the current threshold.

\subsection{Overlap Estimation}
\subsubsection{Single Merge}

We first derive the change in overlap produced by a single additional merge applied to an arbitrary active circuit. Thus, let \(G\) be the current active circuit at some stage of the approximate procedure, and let \(\widetilde G\) be the circuit obtained from \(G\) by performing one further merge. We denote the corresponding prepared states by
$
\ket{\phi}=G\ket{0}$,
$\ket{\widetilde\phi}=\widetilde G\ket{0}.
$
The formula below, therefore, gives the overlap between two consecutive states in the merging procedure (for both neighboring merges and strippings).

Assume that the merge is performed in layer \(k\). Let
\[
\ket{\phi^{(k)}}=\sum_{b\in\{0,1\}^k}\alpha_b\ket{b}
\]
be the coarse-grained state entering that layer in the current circuit \(G\) with coefficients $\alpha_b := \langle b | \phi^{(k)} \rangle$. For a control pattern \(x\in\{0,1,e\}^k\), define
\begin{equation*}
    P_x^{(G)}
    :=
    \sum_{b\in B(x)} |\alpha_b|^2 ,
\end{equation*}
the probability related to the control pattern $x$ of the current circuit $G$.
\begin{proposition}[Single-merge overlap]
\label{prop:single-merge}
Let \(x\) and \(y\) be the two control
patterns that are merged into an arbitrary angle $\theta_C$. Then, the overlap between the state before and after the merge is
\begin{equation*}
    \braket{\,\widetilde\phi|\phi\,}
    =
    1-
    \sum_{w\in\{x,y\}}
    \left(
    1-\cos\!\frac{\theta_w-\theta_C}{2}
    \right)
    P_w^{(G)},
\end{equation*}
\end{proposition}
\begin{proof}
All layers after the merging layer are identical in the two circuits and cancel in the overlap. The same holds for all preceding layers, which prepare the coarse-grained state \(\ket{\phi^{(k)}}\). Hence only the merged layer contributes. For \(b\notin B(z)\), the angle is unchanged, and the contribution is \(1\). For \(b\in B(x)\) or \(b\in B(y)\), the overlap factor is
\[
\bra{0}R_y(-\theta_C)R_y(\theta_w)\ket{0}
=
\cos\!\left(\frac{\theta_w-\theta_C}{2}\right).
\]
Summing over all basis states \(b\in\{0,1\}^k\) gives the result.
\end{proof}

\subsubsection{Repeated Merges}
Prop.~\ref{prop:single-merge} is exact, but it is local: it requires the probabilities \(P_x^{(G)}\) of the \emph{current} circuit at the moment when the merge is performed. In the implementation, these probabilities are not recomputed after every accepted merge. Instead, we evaluate them once from the fixed baseline circuit and then update the overlap estimate using the clusters introduced above. This yields the estimator used in step~(6) of the algorithm.

For each active cluster \(C\), let \(\theta_x\) denote the original baseline angle associated with \(x\in C\), and let \(\theta_C\) be the current angle of the active gate representing that cluster. We define
\begin{equation*}
\label{eq:cluster_loss}
    L_C
    :=
    \sum_{x\in C}
    \left(
    1-\cos\!\left(\frac{\theta_x-\theta_C}{2}\right)
    \right)P_x ,
\end{equation*}
the overlap loss related to the cluster $C$, where $P_x := P_x^{(0)}$ denotes the probability weight precomputed from the baseline circuit. The overlap estimate used by the implementation is
\begin{equation*}
    \mathcal{F}_{\mathrm{est}}
    =
    1-\sum_{C\in\mathcal{C}_{\mathrm{act}}}L_C .
\end{equation*}
When two active clusters \(A\) and \(B\) are merged into \(C=A\cup B\), the estimator is updated as
\begin{equation*}
    \mathcal{F}_{\mathrm{est}}'
    =
    \mathcal{F}_{\mathrm{est}}+L_A+L_B-L_C.
\end{equation*}
In the case of a control-stripping move, only one old cluster is present, so
\begin{equation*}
    \mathcal{F}_{\mathrm{est}}'
    =
    \mathcal{F}_{\mathrm{est}}+L_A-L_C.
\end{equation*}

For a fixed cluster \(C\), the merged angle \(\theta_C\) minimizes \(L_C\), i.e., maximizes
\begin{equation*}
    \sum_{x\in C}P_x\cos\!\left(\frac{\theta_x-\theta_C}{2}\right).
\end{equation*}
Writing
\begin{align*}
    X_C:=\sum_{x\in C}P_x\cos\!\frac{\theta_x}{2}, \qquad
    Y_C:=\sum_{x\in C}P_x\sin\!\frac{\theta_x}{2},
\end{align*}
the optimal merged angle is
\begin{equation}\label{eq:theta_C}
    \theta_C
    =
    2\arg(X_C+iY_C).
\end{equation}
This follows from the identity 
\begin{align*}
    A\cos\alpha\hspace*{-2pt} +\hspace*{-2pt} B\sin\alpha =\sqrt{A^2\hspace*{-2pt}+\hspace*{-2pt}B^2}\cos(\alpha\hspace*{-2pt} -\hspace*{-2pt} \arg(A\hspace*{-2pt}+\hspace*{-2pt}iB)).
\end{align*}

\begin{proposition}[Exactness within one layer]
If all merges are performed within a single layer, then the estimator is exact:
\begin{equation*}
    \mathcal{F}_{\mathrm{est}}=\braket{\tilde{\psi}|\psi}.
\end{equation*}
\end{proposition}

\begin{proof}
Under the single-layer restriction, circuits $G$ and $\widetilde{G}$ coincide on layers $1,\dots k-1$, so the coarse-grained state $\ket{\phi^{(k)}}$ entering layer $k$ equals the baseline state $\ket{\psi^{(k)}}$, giving $P^{(G)}_x = P_x$. Applying Prop.~\ref{prop:single-merge} to each active cluster and summing yields $\langle \widetilde{\psi} | \psi \rangle = 1 - \sum_C L_C = \mathcal{F}_{\mathrm{est}}$.
\end{proof}

\subsection{A Lower Bound on the overlap}

The estimator \(\mathcal{F}_{\mathrm{est}}\) is exact when all merges occur within a single layer, but in general it is not exact across different layers, since earlier merges modify the amplitudes entering later ones. We therefore introduce a lower bound on the overlap, which can be computed at the end of the algorithm.

For a prefix \(x=i_1\cdots i_k\in\{0,1\}^k\), define the amplification factor at the entrance of layer \(k\) by
\begin{equation*}\label{eq:Rk_def}
    R_x^{(k)}
    :=
    \prod_{l=1}^{k}
    \frac{
        g\!\left(\tilde\theta_{i_1\cdots i_{l-1}},\, i_l\right)
    }{
        g\!\left(\theta_{i_1\cdots i_{l-1}},\, i_l\right)
    },
\end{equation*}
where
\[
g(\theta,0)=\cos(\theta/2),
\qquad
g(\theta,1)=\sin(\theta/2),
\]
\(\theta\) denotes the baseline angle, and \(\tilde\theta\) the angle in the final active circuit. For a final active cluster \(C\) in layer \(k\), we set
\begin{equation*}\label{eq:RC_def}
    R_C:=\max_{x\in C}R_x^{(k)}.
\end{equation*}

\begin{theorem}[Lower bound]
\label{thm:rigorous_bound}
Let \(\mathcal{C}_{\mathrm{act}}\) be the set of final active clusters. Then the overlap between the target state \(\ket{\psi}\) and the final prepared state \(\ket{\tilde\psi}\) satisfies
\begin{equation*}\label{eq:rigorous_bound}
    \braket{\tilde\psi|\psi}
    \ge
    1-\sum_{C\in\mathcal{C}_{\mathrm{act}}}R_C\,L_C.
\end{equation*}
\end{theorem}

\begin{proof}
For a concrete prefix \(x\in\{0,1\}^k\), let
\[
H_x:=\sum_{b\in\{0,1\}^{n-k}}\psi_{xb}\tilde\psi_{xb}
\]
be the exact overlap contribution of the subtree rooted at \(x\). In particular, $
H_{\varnothing}=\braket{\tilde\psi|\psi}$.

We prove by backward induction on the remaining depth that
\begin{equation}\label{eq:Hx_claim}
    H_x
    \ge
    P_x^{(k)}R_x^{(k)}
    -
    \sum_{C\subseteq \mathcal{T}(x)} R_C\,L_C(x),
\end{equation}
where $\mathcal{T}$ is the subtree starting from $x$, and \(L_C(x)\) denotes the contribution to \(L_C\) coming from source prefixes of \(C\) that extend \(x\).

If \(x\) is a leaf, then \(H_x=\psi_x\tilde\psi_x=P_x^{(n)}R_x^{(n)}\), and there are no clusters below \(x\), so \eqref{eq:Hx_claim} holds.

Otherwise the descendants of \(x\) split into the two subtrees rooted at \(x0\) and \(x1\),
$
H_x=H_{x0}+H_{x1}$.
Applying the induction hypothesis to \(x0\) and \(x1\), we obtain
\begin{align*}
H_x
&\ge
P_{x0}^{(k+1)}R_{x0}^{(k+1)}
+
P_{x1}^{(k+1)}R_{x1}^{(k+1)}
\nonumber\\
&
-
\sum_{C\subseteq \mathcal{T}(x0)} R_C\,L_C(x0)
-
\sum_{C\subseteq \mathcal{T}(x1)} R_C\,L_C(x1).
\end{align*}
By the definitions of \(P_x^{(k)}\) and \(R_x^{(k)}\),
\begin{align*}
P_{x0}^{(k+1)}R_{x0}^{(k+1)}
&=
P_x^{(k)}R_x^{(k)}
\cos\!\frac{\theta_x}{2}
\cos\!\frac{\tilde\theta_x}{2},\\
P_{x1}^{(k+1)}R_{x1}^{(k+1)}
&=
P_x^{(k)}R_x^{(k)}
\sin\!\frac{\theta_x}{2}
\sin\!\frac{\tilde\theta_x}{2}.
\end{align*}
Therefore,
\begin{align*}
H_x
&\ge
P_x^{(k)}R_x^{(k)}
-
P_x^{(k)}R_x^{(k)}
\left(
1-\cos\!\frac{\theta_x-\tilde\theta_x}{2}
\right)
\nonumber\\
&
-
\sum_{C\subseteq \mathcal{T}(x0)} R_C\,L_C(x0)
-
\sum_{C\subseteq \mathcal{T}(x1)} R_C\,L_C(x1).
\end{align*}

If \(x\) is unchanged in the final circuit, then the cosine term is zero and \eqref{eq:Hx_claim} follows immediately. Otherwise, \(x\) belongs to a unique final active cluster \(C_x\) with active angle \(\theta_{C_x}=\tilde\theta_x\). The corresponding contribution to the cluster loss is
\[
L_{C_x}(x)
=
P_x^{(k)}
\left(
1-\cos\!\left(\frac{\theta_x-\theta_{C_x}}{2}\right)
\right).
\]
Since \(R_x^{(k)}\le R_{C_x}\), we have
\[
P_x^{(k)}R_x^{(k)}
\left(
1-\cos\!\left(\frac{\theta_x-\tilde\theta_x}{2}\right)
\right)
\le
R_{C_x}\,L_{C_x}(x).
\]
Substituting this into the previous inequality gives \eqref{eq:Hx_claim} for \(x\).

Finally, taking \(x=\varnothing\), we have \(H_{\varnothing}=\braket{\tilde\psi|\psi}\), \(P_{\varnothing}^{(0)}=1\), \(R_{\varnothing}^{(0)}=1\), and \(L_C(\varnothing)=L_C\). Hence \eqref{eq:Hx_claim} becomes
\[
\braket{\tilde\psi|\psi}
\ge
1-\sum_{C\in\mathcal{C}_{\mathrm{act}}}R_C\,L_C,
\]
which proves the theorem.
\end{proof}

\begin{figure}
    \centering
    \includegraphics[width=1.\linewidth]{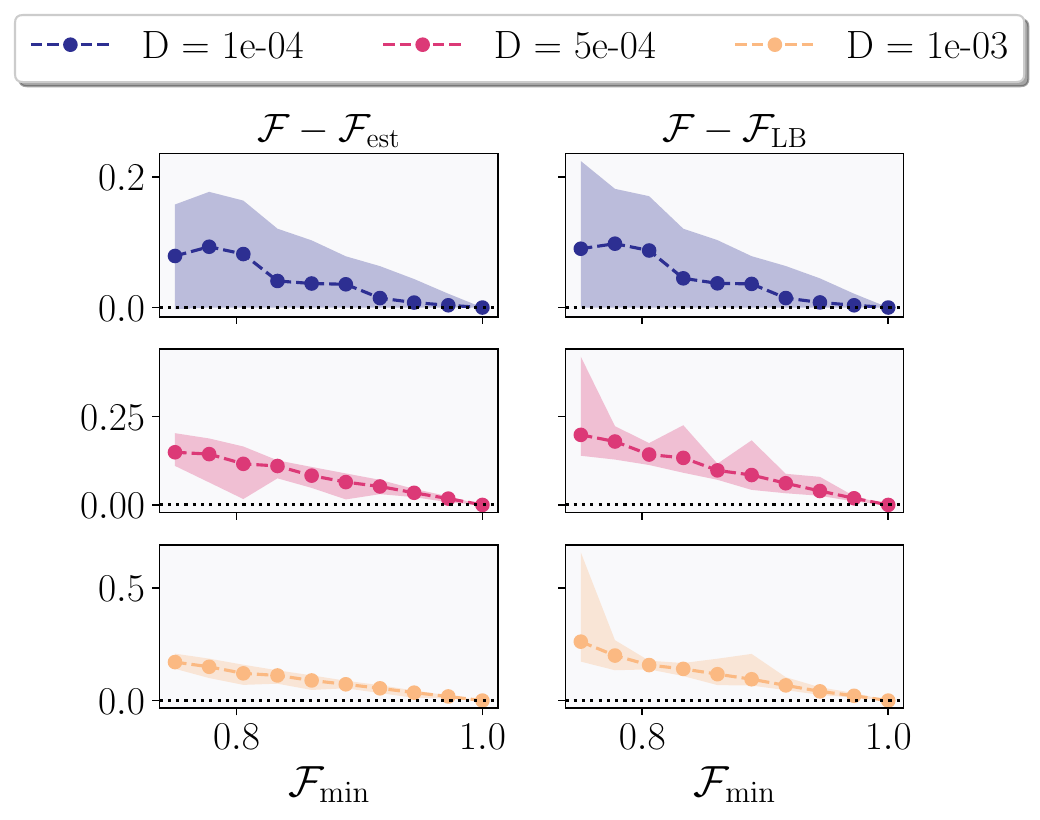}
    \caption{
    Performance assessment of the overlap estimation. 
    The number of qubits is fixed to $15$, and the number of intervals to $M = 20$, while the sparsity percentage $D$ is varied. 
    The two columns show, respectively, 
\(\mathcal{F}_{\mathrm{est}}-\mathcal{F}\) and 
\(\mathcal{F}-\mathcal{F}_{\mathrm{LB}}\) as functions of the minimum allowed overlap \(\mathcal{F}_{\min}\), where $\mathcal{F}_{\mathrm{LB}}$ is the lower bound  of Thm.~\ref{thm:rigorous_bound}, and $\mathcal{F}$ is the true overlap.
    The curves represent averages over 20 random real instances, while the shaded areas show the minimum and maximum values observed across repetitions.
    }
    \label{fig:square-root-fidelity-estimate}
\end{figure}

\begin{figure}
    % change font size
    % write "# CNOTs approx / # CNOTs
    % change to D
    \centering
    \includegraphics[width=\linewidth]{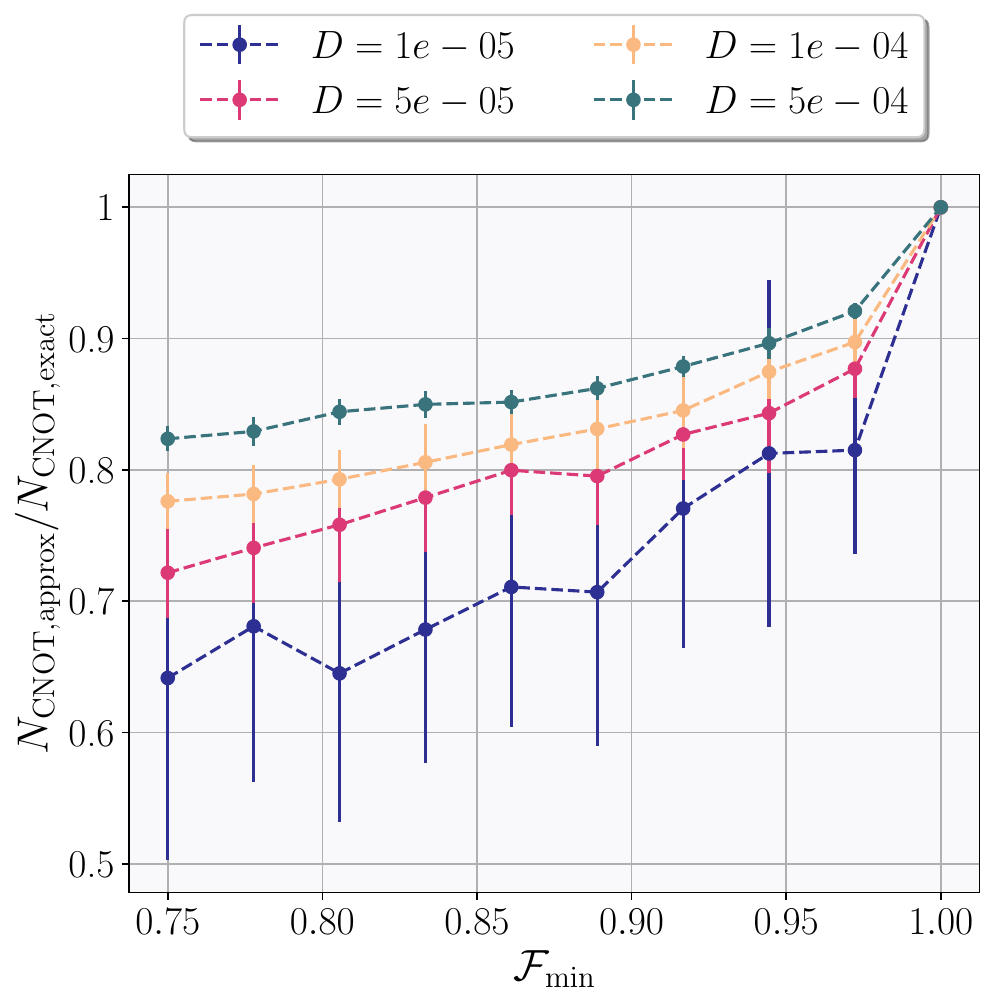}
    \caption{
    Comparison between the exact and approximate algorithms. 
    The number of qubits is fixed to $n = 20$, and the number of intervals to $M=20$, while the sparsity percentage $D$ is varied.  
    The plot shows the ratio of the number of CNOTs in the approximate algorithm to that in the exact scheme. 
    Each data point corresponds to an average over 20 repetitions, with 10 points plotted per curve.
    }
    \label{fig:approx_vs_hybrid}
\end{figure}
\begin{figure}
    \centering
    \includegraphics[width=\linewidth]{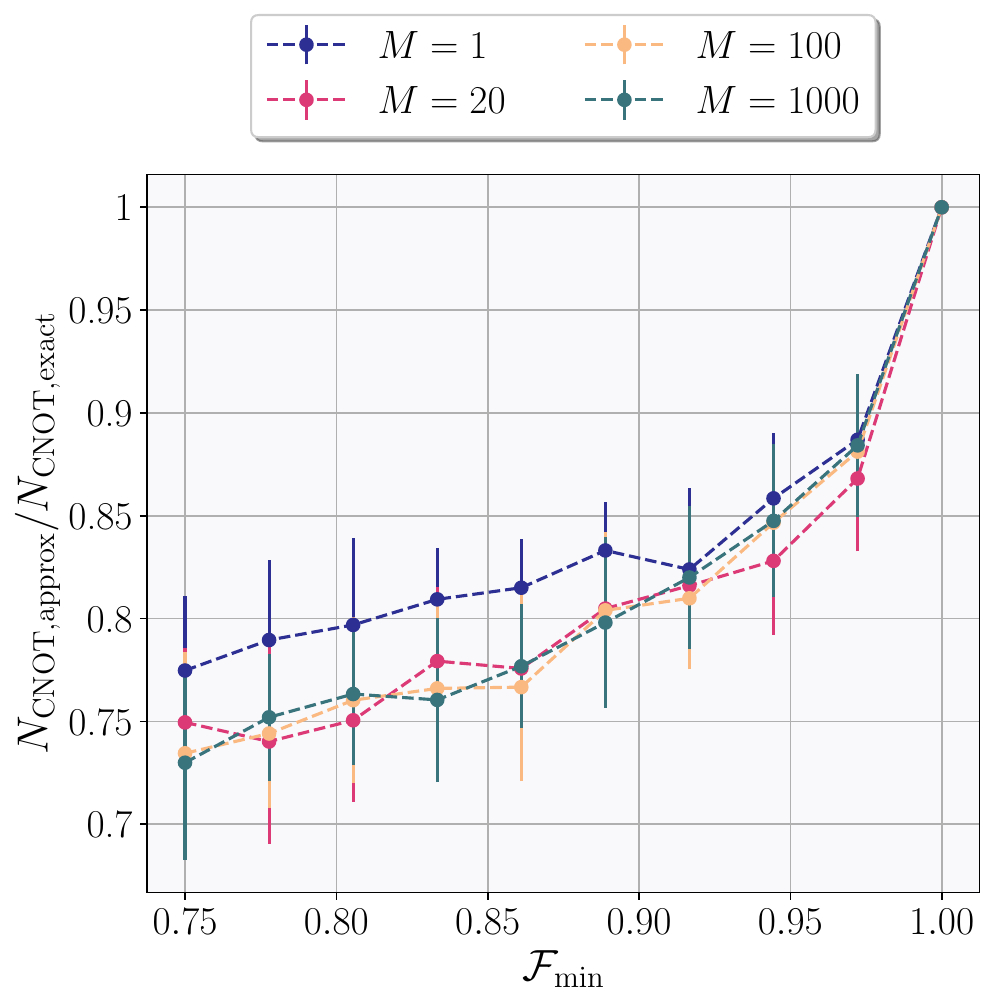}
    \caption{Comparison between the exact and the approximate algorithm for different values of $M$. The number of qubits is fixed to $n = 20$, and the sparsity percentage to $D = 5\cdot 10^{-5}$. The plot shows the ratio of the number of CNOTs in the approximate algorithm to that in the exact scheme. 
    Each data point corresponds to an average over 20 repetitions, with 10 points plotted per curve.}
    \label{fig:ratio_M_sweep}
\end{figure}

\subsection{Classical complexity}
At the beginning, one applies the improved exact optimization of Sec.~\ref{sec:exactoptimization}, which has complexity
$
\mathcal{O}(d^2n^3 + dn^2\log d).
$Then, the approximate routine is performed.

The nonzero leaves of the baseline circuit and their prefix-probability tables are precomputed once. Since there are at most $d$ supported leaves and $n$ depths, this preprocessing costs $\mathcal{O}(dn)$.

At a given layer $k$, there are at most $d$ active gates. For each gate, there are at most $k$ neighbor candidates and at most $k$ control-removal candidates, so the total number of candidates in one pass is bounded by
$
\sum_{k=1}^{n}\mathcal{O}(dk)=\mathcal{O}(dn^2).
$

For each candidate, the implementation constructs the merged cluster by scanning only the supported concrete prefixes at that depth together with the active keys in the corresponding layer. In the sparse regime, both of these have size at most $d$, so the cost per candidate is $\mathcal{O}(d)$. Hence, generating and evaluating all candidates in one pass costs $
\mathcal{O}(d^2n^2)$.
Sorting the candidate list adds
$
\mathcal{O}\!\bigl(dn^2\log(dn^2)\bigr),
$
which is lower order than the candidate-evaluation term.

Let $T$ be the number of greedy passes, which is given by the $M$ threshold steps plus the final saturation passes at the end. Since every successful merge reduces the number of active gates, and there are at most $\mathcal{O}(dn)$ active gates overall, one has the worst-case bound
$
T=\mathcal{O}(M+dn).
$
Therefore the approximate stage has worst-case complexity
$
\mathcal{O}\!\bigl((M+dn)\,d^2n^2\bigr).
$

Including the preprocessing, the total complexity is
$ \mathcal{O}\!\bigl(d^2n^3 + dn^2\log d + (M+dn)\,d^2n^2\bigr). $

In practice, this estimate is often pessimistic, since many candidates are discarded immediately by the overlap threshold.

If one also evaluates the rigorous lower bound of Thm.~\ref{thm:rigorous_bound} at the end of the approximate routine, this adds only a lower-order postprocessing cost. Indeed, the implementation already maintains the baseline support, the active clusters, and their losses during the greedy procedure. The remaining task is to compute the amplification factors \(R_x^{(k)}\) for all supported leaves and all depths, and then to extract the corresponding maxima \(R_C\) for the final active clusters.

Since there are at most \(d\) supported leaves and \(n\) depths, building the full amplification table costs \(\mathcal{O}(d^2 n)\): for each leaf and depth, one has to identify the matching baseline and active gates in layers of size at most \(d\). The subsequent scan over the final active clusters also costs \(\mathcal{O}(d^2 n)\), because the total number of concrete source patterns stored across all clusters is \(\mathcal{O}(dn)\), and for each such pattern one compares against at most \(d\) supported leaves. Therefore, the total cost of computing the lower bound is
$
\mathcal{O}(d^2 n),
$
which does not change the overall asymptotic scaling.

\subsection{Numerical Results}

First, we aim to assess the performance of our overlap estimation on practical numerical instances. Let us indicate the overlap between the target state $\ket{\psi}$ and the prepared state $\ket{\tilde{\psi}}$  by $\mathcal{F}$. Fig.~\ref{fig:square-root-fidelity-estimate} shows, as a function of $\mathcal{F}_{min}$, the difference between our estimation $\mathcal{F}_{est}$ and $\mathcal{F}$ on the first column, and between $\mathcal{F}$ and its lower bound of Thm. \ref{thm:rigorous_bound} on the second column, for different values of the sparsity percentage $D := d / 2^n$. \newline
The numerical results confirm that the estimate closely tracks the true overlap across random instances and behaves almost always as a lower bound. Since it yields a more accurate estimate of the overlap than the lower bound and has lower classical complexity, it justifies its use in the approximate algorithm. We also note that the estimation is more accurate for lower values of $\mathcal{F}_{min}$ and lower sparsity, which supports our implementation.

We then compared the approximate algorithm with the exact approach in Fig.~\ref{fig:approx_vs_hybrid}, plotting the ratio of their relative CNOT counts and averaging over random real instances. \newline
The numerical results show that allowing a small, controlled approximation further reduces the number of gates by about $20-30\%$. The improvement is most pronounced for sparser target states, where, for certain instances, it reaches $50\%$.

Finally, we studied the dependence on the number of intervals $M$. Fig. \ref{fig:ratio_M_sweep} shows the ratio of the number of CNOTs in the approximate algorithm to that in the exact scheme, with a fixed sparsity percentage, and varying $M$. We observe that beyond $M=20$, increasing this parameter doesn't yield a significant difference in the reduction of CNOTs.

\section{Conclusions}

We studied the Grover-Rudolph state-preparation algorithm in the sparse regime, with the goal of reducing the circuit size. First, we introduced an exact optimization procedure that combines the sparse merging strategy with uniformly controlled rotations, selecting the cheaper implementation at each layer. 
This shows an improvement for sparse states, reaching a $90\%$ reduction compared to the unoptimized Grover--Rudolph algorithm, when the sparsity percentage is $10^{-5}$.
We then proposed an approximate merging strategy that selects merges while enforcing a minimum target overlap. We derived an explicit formula for the reduction in overlap that can be computed classically starting from the rotation angles. This yielded a further reduction of $20-30\%$ in the number of CNOTs compared to the exact approach. 

There are several natural directions for future work. The most immediate one is to extend the analysis beyond real states and treat the fully complex case, where approximate merging must also account for phase gates. It would also be interesting to combine our approach with more advanced decomposition methods, in particular ancilla-assisted constructions, and to study whether similar approximation ideas can improve fault-tolerant cost measures such as Toffoli count. Furthermore, there could be other kinds of possible merges that one could allow, for example, those that we called regional merges. Thus, it would be interesting to generalize our estimation to allow more operations on the gates.

\section{Code and data availability}
The code and data are available at \href{https://github.com/Damuna/approx-grover-rudolph}{approx-grover-rudolph}. All the results were obtained using Python.

\section{Acknowledgments}
This work was supported by the QDFG under Germany’s Excellence Strategy - EXC-2123 QuantumFrontiers-2 - 390837967 and SFB 1227 (DQ-mat), the Quantum Valley Lower Saxony, and the BMBF projects ATIQ, SEQUIN, Quics, and CBQD. Helpful correspondence and discussions with Markus Frembs and Tobias J. Osborne are gratefully acknowledged.

% IF YOU WANT TO CHECK OLD NOTES
%\input{old.tex}

\bibliographystyle{unsrt}
\bibliography{bibliography}
\end{document}